\documentclass[aps, prl, twocolumn, superscriptaddress]{revtex4}
\usepackage{hyperref}
\usepackage{verbatim}
\usepackage{amsmath}
\usepackage{latexsym}
\usepackage{revsymb}
\usepackage{yfonts}

\usepackage{natbib}
\usepackage{amsfonts}
\usepackage{amsmath}
\usepackage{amssymb}
\usepackage{amsthm}
\usepackage{graphicx}
\usepackage{bm}
\usepackage{bbm}

\newcommand{\be}{\begin{eqnarray} \begin{aligned}}
\newcommand{\ee}{\end{aligned} \end{eqnarray} }
\newcommand{\benn}{\begin{eqnarray*} \begin{aligned}}
\newcommand{\eenn}{\end{aligned} \end{eqnarray*} }

\newcommand{\bc}{\begin{center}}
\newcommand{\ec}{\end{center}}

\newcommand{\id}{\mathbb{I}}

\newcommand{\tr}{\mathop{\mathrm{tr}}\nolimits}

\newcommand{\hil}{\mathcal{H}}

\newtheorem{lemmaApp}{Lemma A.}
\newtheorem{theoremAppB}{Theorem B.}
\newtheorem{corollaryApp}{Corollary A.}
\newtheorem{theoremApp}{Theorem A.}
\newtheorem{lemmaAppC}{Lemma C.}

\usepackage{amsfonts}

\def\id{\mathbb{I}}

\def\01{\{0,1\}}

\newcommand{\eps}{\varepsilon}

\newcommand{\proj}[1]{|#1\rangle\langle#1|}
\newcommand{\ketbra}[2]{|#1\rangle\langle#2|}

\newcommand{\SH}{{\rm H}}
\newcommand{\hmin}{{\rm H_{\rm min}}}
\newcommand{\halpha}{{\rm H_{\alpha}}}
\newcommand{\hmineps}{{\rm H}^\eps_{\rm min}}

\begin{document}

\title{A min-entropy uncertainty relation for finite size cryptography}
\author{Nelly Ng Huei Ying}
\email[]{nell0002@e.ntu.edu.sg}
\affiliation{Centre for Quantum Technologies, National University of Singapore, 3 Science Drive 2, Singapore 117543}
\affiliation{School of Physical and Mathematical Sciences, Nanyang Technological University, 21 Nanyang Link, Singapore 637371}
\author{Mario \surname{Berta}}
\email[]{berta@phys.ethz.ch}
\affiliation{Institute for Theoretical Physics, ETH Zurich, 8093 Zurich, Switzerland}
\author{Stephanie Wehner}
\email[]{wehner@nus.edu.sg}
\affiliation{Centre for Quantum Technologies, National University of Singapore, 3 Science Drive 2, Singapore 117543}
\date{\today}

\begin{abstract}
Apart from their foundational significance, entropic uncertainty relations play a central role in proving the security
of quantum cryptographic protocols. Of particular interest are thereby relations in terms of the smooth
min-entropy for BB84 and six-state encodings. Previously, strong uncertainty relations were obtained which are valid
in the limit of large block lengths. Here, we prove a new uncertainty relation in terms of the smooth min-entropy that is
only marginally less strong, but has the crucial property that it can be applied to rather small block lengths. This paves the way 
for a practical implementation of many cryptographic protocols. As part of our proof we show tight uncertainty relations for a family of R{\'e}nyi entropies that may be of independent interest. 
\end{abstract}

\maketitle

Entropic uncertainty relations form a modern way to characterize the uncertainty inherent in several quantum measurements. As opposed to more traditional methods of capturing the notion of uncertainty, they have the advantage that they are able to quantify uncertainty solely in terms of the measurements we consider, and are independent of the state to be measured. To see this clearly, let us explain the notion of entropic uncertainty in more detail (also, see~\cite{ww:ursurvey} for a survey). Suppose we are given a state $\rho$ on which we can make one of $L$ possible measurements with outcomes labelled $x \in \mathcal{X}$. Let $p_{x|\rho,\theta}$ denote the probability of observing outcome $x$ when making the measurement labelled $\theta$ on the state $\rho$. For each measurement, we can consider some form of entropy of the outcome distribution such as for example the Shannon entropy $\SH(X|\Theta=\theta) = - \sum_x p_{x|\rho,\theta} \log_2 p_{x|\rho,\theta}$.
An entropic uncertainty relation in terms of the Shannon entropy is then determined by the average ($p_\theta = 1/L$) over the individual entropies. More precisely, such a relation states that for \emph{all} states $\rho$

\begin{align}\label{eq:shannonur}
\frac{1}{L} \sum_{\theta} \SH(X|\Theta=\theta) = \SH(X|\Theta) \geq c\ ,
\end{align}

where $c$ is a constant that depends solely on the measurements. For example, if $\rho$ is a single qubit state, and we consider $L=2$ measurements in the Pauli $\sigma_X$ and $\sigma_Z$ eigenbases, we have $c=\frac{1}{2}$ \cite{Maassen88}. To see why~\eqref{eq:shannonur} for $c > 0$ is indeed connected with uncertainty, note that if the outcome is certain with respect to some measurement $\theta$ on the state $\rho$ ($\SH(X|\Theta = \theta) = 0$), then the outcome of at least one other measurement $\theta' \neq \theta$ is uncertain ($\SH(X|\Theta = \theta') > 0$). Similarly, the larger the value of $c$, the more uncertain these outcomes are. The value of $c$ thus give a natural measure of the incompatibility of different sets of measurements. \emph{Strong} uncertainty relations have the property that $c$ is large. 

From a cryptographic perspective, uncertainty relations in terms of the \emph{min-entropy} $\hmin(X|\Theta=\theta) = - \log \max_x p_{x|\rho,\theta}$ \footnote{We use $\log$ throughout the paper as the base 2 logarithm, unless otherwise stated.} are of particular interest, since the min-entropy determines how many random bits (key) can be extracted from $X$~\cite{renato:operational}. 
In a cryptographic setting, it is thereby often interesting to consider a slight extension of the notion of uncertainty relations above. Namely, instead of measuring one state $\rho$, we imagine that an adversary prepares with some probability $p_k$ a state $\rho_k$ (labelled by some classical label $K=k$) which we subsequently measure. Since entropic uncertainty relations hold for any state, they do in particular hold for any state $\rho_k$ that the adversary may have prepared. Yet, the distribution $\{p_{x|k\theta}\}$ over measurement outcomes may of course depend on $k$. 
Uncertainty relations with respect to such classical side information $K$ thus take the form
\begin{align}
\hmin(X|\Theta K) \geq c'\ ,
\end{align}
for some constant $c'$ depending on the measurements we make. Averaging over bases $\Theta$ and classical information K, the conditional min-entropy is given by (see appendix)
\begin{align}
\hmin (X|\Theta K) = -\log \sum_\theta p_\theta \sum_k p_{k|\theta} \max_x ~ p_{x|k\theta}\ .
\end{align}
For example, imagine that $\rho$ is an $n$-qubit state and we perform one of the $2^n$ possible measurements given by measuring each qubit independently in one of the two BB84 bases~\cite{bb84}, i.e., in the eigenbasis of Pauli $\sigma_x$ or $\sigma_z$. It is known that in this case $c' = -n\cdot\log(1/2 + 1/(2\sqrt{2}))\approx n\cdot0.22$ for any $K$. This is also optimal as there exists a state that attains this lower bound. 

Measurements in BB84 bases are indeed common in many quantum cryptographic protocols. In particular, they are used in two-party cryptographic protocols in the bounded~\cite{serge:bounded,serge:new} and noisy-storage model~\cite{Noisy1,noisy:robust, noisy:new}. These models allow for the secure implementation of any two-party cryptographic primitive under the assumption that the adversary's quantum memory device is bounded and imperfect. This includes interesting primitives such as oblivious transfer, bit commitment, and even secure identification of e.g.~a user to an ATM machine. The security of all protocols in this model crucially rests on the existence of uncertainty relations in terms of min-entropy~\cite{serge:bounded,serge:new,secureid,Noisy1,noisy:robust,noisy:new,qcextract}. Yet, the value of $c' \approx n\cdot0.22$ for BB84 bases is usually too low to be cryptographically useful. In particular, a low value for $c'$ means that the adversary's memory must be very limited and/or noisy for security to be possible~\cite{serge:bounded,serge:new,noisy:new} at all. Furthermore, a low value of $c'$ means that any experiment implementing such protocols can tolerate only a small amount of bit flip errors and losses~\cite{noisy:robust, Curty10,secureid:practical}. For instance, if $p_{\rm err}$ is the bit flip error on the channel connecting Alice and Bob, then security for the cryptographic primitive known as oblivious transfer is possible if $c' - h(p_{\rm err}) > 0$~\cite{chris:diss,Curty10}, where $h(p) = - p \log_2 p - (1-p) \log_2 (1-p)$ is the binary Shannon entropy. 

Motivated by this need to obtain a strong uncertainty relation for BB84 bases, that is, a large $c'$, the authors of~\cite{serge:new} considered the so-called \emph{smooth} min-entropy $\hmineps(X|\Theta K)$. Intuitively, a lower bound $c'$ on this quantity tells us that we do indeed have min-entropy at least $c'$, except for some small error parameter $\eps>0$. Formally, this quantity is defined as (see appendix)
\begin{align}
\hmineps(X|\Theta K)_\rho = \sup_{\rho'}~\hmin(X|\Theta K)_{\rho'}\ ,
\end{align}
where $\rho'$ is $\epsilon$-close to $\rho$ in terms of the purified distance~\cite{Tomamichel09}.

It turns out that at the expense of such a small error $\eps$, a much stronger uncertainty relation can indeed be obtained. In particular, it has been shown~\cite{serge:new} that for measurements in the BB84 bases and any $\delta\in(0,\frac{1}{2}]$,
\begin{align}\label{eq:bb84old}
	\hmineps(X|\Theta K) \geq n\cdot\left(\frac{1}{2} - \delta\right)\ ,
\end{align}
where
\begin{align}\label{oldeps}
\eps = \exp \left[-\frac{\delta^2n}{128(2+\log\frac{2}{\delta})^2}\right]\ .
\end{align}
Using this relation in a cryptographic protocol only yields an additional error $\eps$ in the overall security error, and it is widely employed in the protocols of~\cite{serge:new,secureid,chris:diss,noisy:new,Curty10,secureid:practical}. 

From a theoretical (asymptotic) viewpoint, this uncertainty relation is certainly sufficient. Yet, when it comes to putting any of such protocols into a practical experiment it has a small caveat: whereas $\eps$ decreases exponentially in the number of qubits $n$, for a large amount of uncertainty, i.e., $c'=1/2-\delta\approx1/2$, the convergence is extremely slow. For example, for $\delta = 0.0106$~\cite{secureid:practical} corresponding to $c' =0.4894$, we need $n \geq 2.39 \times 10^8$ to even have $\eps = 0.1$! In an experiment using weak coherent pulses, with frequency of 1GHz and Poisson parameter $\mu=1$ it takes approximately $2.5$ seconds to generate such an $n$~\cite{secureid:practical} if there are absolutely no losses of any kind. However, compared to the generation time, a more significant inconvenience is that the classical post-processing of such large block lengths is time-consuming. 


\section{Results}

To implement aforementioned protocols, it would thus be desirable to have a relation that is useful for significantly smaller values of $n$. Here, we prove such a relation that makes a statement for any desirable \emph{fixed} error $\eps>0$. In particular, we show that for any $n$ qubit quantum state $\rho$ and measurements in BB84 bases
\begin{align}
	\hmineps(X|\Theta K) \geq n\cdot c_{BB84}\ ,
\end{align}
where
\begin{align}\label{eq:bb84result}
	c_{BB84} := \max_{s\in(0,1]} \frac{1}{s} \left[1+s-\log(1+2^s) \right] - \frac{1}{sn} \log \frac{2}{\epsilon^2}\ .
\end{align}

At the first glance, it may be hard to see that $c_{BB84}$ is indeed large. However, applying it to the example from~\cite{secureid:practical} (see above) by plugging in $s=0.1$ demonstrates that for the same $\eps = 0.1$, $c_{BB84} \geq 0.4894$ whenever $n\geq 2.36\times 10^4$. Comparing this with calculations in the previous section, the required block length $n$ is approximately $10^{-4}$~times smaller. Figure~\ref{fig:plot} provides a comparison of these two bounds. We see that even for large $\epsilon$, the required bound on the block length $n$ given by~\eqref{oldeps} is large.

\begin{figure}[h!]\label{fig:plot}
\includegraphics[scale=0.49]{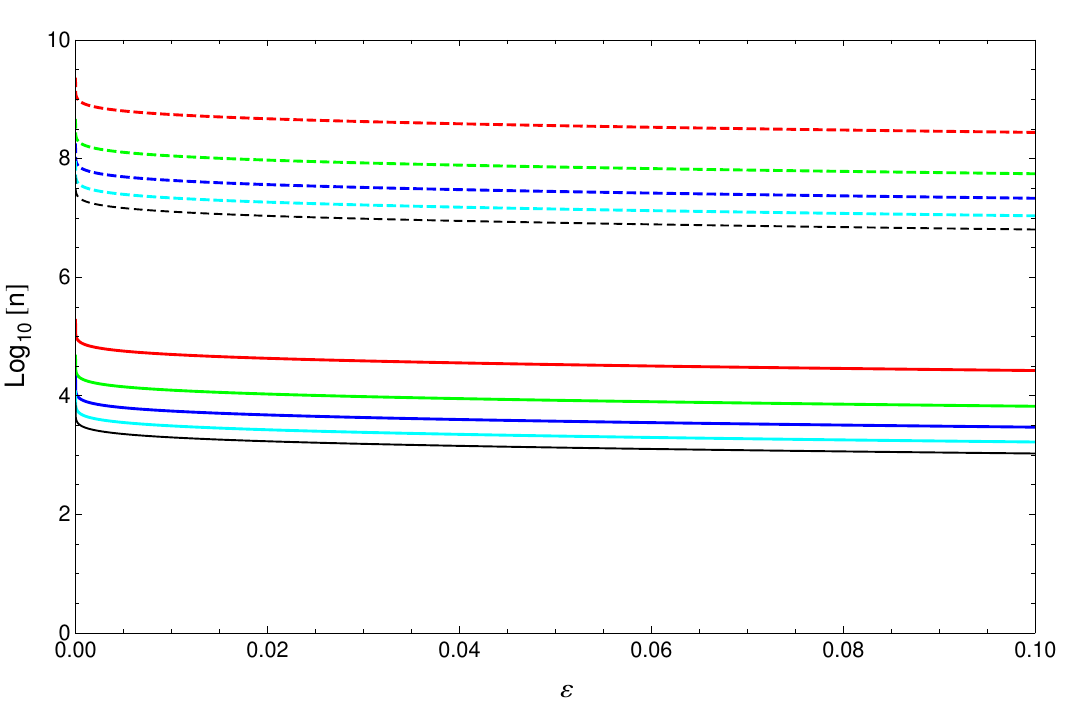}
\caption{(Color online) This plot shows the minimal required block length $n$ on a logarithmic scale of base $10$, in order to achieve an error parameter $\epsilon$. The dashed curves are plotted for the previous known bound~\eqref{oldeps}, while the solid lines are obtained from our new analysis~\eqref{eq:bb84result}. The different colors represent the fixed values of the lower bound $c'$, with values 0.45, 0.46, 0.47, 0.48, and 0.49 respectively. As $c'$ increases, the plotted bounds get relatively higher.}
\end{figure}

Our relation can readily be applied to any BB84 based two-party protocols in the bounded (or noisy)-storage model, and enables experiments for significantly smaller values of $n$. For example, it enables the experimental implementation of~\cite{noisyimpl} with $n = 2.5\times 10^5$ instead of $n > 10^{9}$ for the same error parameter $\eps$.

Furthermore our relation can be extended to the case of six-state protocols, i.e., measurements in Pauli $\sigma_x$, $\sigma_z$ and $\sigma_y$ eigenbases as suggested in~\cite{chris:diss,secureid,qcextract}. For this case we obtain 
\begin{align}
	\hmineps(X|\Theta K) &\geq n\cdot c_6\ ,
\end{align}
where
\begin{align}
	c_6 := \max_{s\in(0,1]} &~-\frac{1}{s} \log\left[  \frac{1}{3} \left(1+2^{1-s}\right) \right] - \frac{1}{sn} \log \frac{2}{\epsilon^2}\ .
\end{align}
This yields a similar improvement over the relation analogous to~\eqref{eq:bb84old} proven in~\cite{serge:new}.

A crucial step in our proof is to show \emph{tight} uncertainty relations for conditional R{\'e}nyi entropies of order $\alpha$, denoted by $\halpha(A|B)$. These may be of independent interest. Previously, such relations were only known for single qudit measurements for $\alpha\rightarrow 1$, $\alpha = 2$, and $\alpha \rightarrow \infty$ (see e.g.~\cite{ww:ursurvey,ww:cliffordUR,renyiArxiv}). More precisely, we show that for measurements on $n$-qubit states $\rho$ in BB84 bases, the minimum values of the conditional R{\'e}nyi entropies for any $\alpha\in(1,2]$ are
\begin{align}\label{eq:alphanqubit}
	\min_\rho \halpha(X|\Theta)_{\rho|\rho} = n\cdot\frac{\alpha-\log(1+2^{\alpha-1})}{\alpha-1}\ ,
\end{align}
where 
\begin{equation}
\halpha(A|B)_{\rho|\rho} :=  \frac{1}{1-\alpha} \tr\left[\rho_{AB}^\alpha (\mathbb{I}_A\otimes \rho_{B})^{1-\alpha}\right]\ .
\end{equation}
Similarly, for measurements in the six-state bases
\begin{align}
	\min_\rho \halpha(X|\Theta)_{\rho|\rho} = n\cdot\frac{\log3-\log\left(1+2^{2-\alpha}\right)}{\alpha-1}\ .
\end{align}


\section{Proof}

Let us now explain the proof of our results. A technical derivation including all details may be found in the appendix. For simplicity, we restrict ourselves to the case of BB84 measurements. An extension for six-state protocols is analogous and can be found in the appendix. To obtain~\eqref{eq:bb84result} we proceed in four steps. First, we will prove a tight uncertainty relation in terms of the $\alpha$-R{\'e}nyi entropy when $\rho$ is just an $n=1$ qubit state. Second, we show how to extend this result to an uncertainty relation for $n > 1$ qubits, giving us~\eqref{eq:alphanqubit}. The third step is to reintroduce $K$ as outlined in the introduction. Finally, we relate the R{\'e}nyi entropies of order $\alpha \in (1,2]$ to the smooth min-entropy. \\

\medskip
\noindent

{\bf Step 1 - A single qubit uncertainty relation:}

For the case when $A$ and $B$ are classical the conditional $\alpha$-R{\'e}nyi entropy reduces to the simple form
\begin{equation}
\halpha (A|B)_{\rho|\rho} = \frac{1}{1-\alpha} \log \sum_{b} p_{B=b} \sum_{a} p^\alpha_{A=a|B=b}\ .
\end{equation}
The relevant $\alpha$-R{\'e}nyi entropy for a single qubit state $\rho_k$ (where k denotes some classical information associated with the state $\rho_k$) is
\begin{align}
\halpha (X|\Theta)_{\rho_k|\rho_k} &=\frac{1}{1-\alpha} \log \sum_{\theta\in\01} p_{\theta} \sum_{x\in\01} p^\alpha_{x|k \theta}\nonumber\\
&=\frac{1}{1-\alpha}\log\left[\frac{1}{2}\cdot\sum_{\theta\in\01,x\in\01} p^\alpha_{x|k \theta}\right]\ .
\end{align}
Here $p_{x|k \theta}:=\tr(M_{x|\theta}~\rho_k)$, where $M_{x|\theta}$ denotes the measurement operator 
\begin{equation}
M_{x|\theta}= \mathbf{H}^\theta|x \rangle\langle x|\mathbf{H}^\theta\ ,
\end{equation}
with $\mathbf{H}$ the Hadamard matrix. To minimize the $\alpha$-R{\'e}nyi entropy for values of $\alpha\in(1,2]$, it it sufficient to maximize the summation term. Defining
\begin{equation}\label{PXT}
P(X|\Theta)_{\rho_k} = \frac{1}{2}\cdot\sum_{\theta\in\01,x\in\01} p^\alpha_{x|k \theta}\ ,
\end{equation}
we first rewrite $p_{x|k \theta}$ as functions of two variables: $g_x:=\tr(\sigma_x\rho_k)$ and $g_z:=\tr(\sigma_z\rho_k)$. The Bloch sphere condition for a qubit gives $g_x^2+g_y^2+g_z^2\leq g_x^2+g_z^2\leq 1$, which serves as a constraint in maximizing \eqref{PXT}. Switching to spherical coordinates and evaluating the partial derivatives of \eqref{PXT} according to multiple independent variables, we prove
\begin{eqnarray}\label{minimal_renyi_bb84}
\halpha(X|\Theta)_{\rho_k|\rho_k} &\geq& \frac{1}{1-\alpha} \log \left[ \frac{1}{2^{1+\alpha}} (2^{\alpha}+2) \right]\nonumber\\
&=& \frac{1}{\alpha-1} \left[\alpha-\log(1+2^{\alpha-1}) \right].
\end{eqnarray}
Moreover, the minimal $\alpha$-R{\'e}nyi entropy is achieved on an eigenstate of either measurement basis.

\medskip
\noindent

{\bf Step 2 - A relation for $n$-qubits:} To extend the one qubit uncertainty relation to multiple qubits, the central problem is to prove that the lower bound on the conditional entropy scales linearly with the block length $n$. This essentially implies that for a system of $n$ qubits, the entanglement across qubits does not give rise to a lower minimal $\alpha$-R{\'e}nyi entropy. In our analysis, we show this by first considering the last qubit measured, conditioned on all the previous $n-1$ measurement bases and values. That is, we consider a $n$-qubit normalized density operator $\rho_{ABk}$, where $B$ denotes the last qubit and $A$ the remaining $n-1$ qubits, and write 
\begin{equation}\label{sumterm}
P(X_B|\Theta)_{\rho_{ABk}} = \frac{1}{2}\cdot\sum_{\theta_B,x_B\in\01} p^\alpha_{x_B|\theta_B  x_A \theta_A k}\ ,
\end{equation}
where $p_{x_B|\theta_B  x_A \theta_A k}=\tr(M_{x_B|\theta_B}\sigma_B)$ with the corresponding normalized density operator
\begin{equation}
\sigma_B =\tr_A \left[ \frac{M_{x_A|\theta_A}~\rho_{ABk}~M_{x_A|\theta_A}^\dagger}{\tr\left[M_{x_A|\theta_A}~\rho_{ABk}~M_{x_A|\theta_A}^\dagger\right]} \right]\ .
\end{equation}
Since the uncertainty relation for one qubit~\eqref{minimal_renyi_bb84} holds for any density operator, it holds in particular for $\sigma_B$. By induction, it is then easily shown that the minimal entropy is additive.

\medskip
\noindent

{\bf Step 3 - Classical side information $K$:} After Steps 1 and 2, we established a tight uncertainty relation for a binary string $X^n$ conditioned on the basis string $\Theta^n$. Namely, we have
\begin{equation}
\halpha(X^n|\Theta^n)_{\rho_k|\rho_k} \geq n\cdot\frac{1}{\alpha-1} \left[\alpha-\log(1+2^{\alpha-1}) \right]\ .
\end{equation}
for any $n$-qubit state $\rho_k$. In this step, we obtain the conditioning with relation to classical side information K. In other words, we need to evaluate $ \halpha(X|\Theta K)_{\rho|\rho}$ with
\begin{equation}
\rho = \sum_{\theta\in\01^n} p_{\theta}~\ketbra{\theta}{\theta} \sum_k p_{k|\theta}~\rho_k \sum_{x\in\01^n} p_{x|\theta k}~\ketbra{x}{x}\ .
\end{equation}
By observing the independence of $\Theta$ and K, we show that the bounds of these values coincide, implying that
\begin{equation}\label{step3end}
\halpha(X|\Theta K)_{\rho|\rho} \geq n\cdot \frac{1}{\alpha-1} \left[\alpha-\log(1+2^{\alpha-1}) \right]\ .
\end{equation}

\smallskip
\noindent

{\bf Step 4 - Relation to the min-entropy:} As motivated previously, the final desired measure of entropy is the \textit{smooth} min-entropy $\hmineps (X|\Theta K)_{\rho}$. A recent work~\cite{QAEP} has shown that a lower bound can be obtained for this quantity. Namely, we have for any state $\rho$ and $\alpha\in(1,2]$
\begin{equation}
\hmineps (X|\Theta K)_{\rho} \geq \halpha (X|\Theta K)_{\rho|\rho} -\frac{1}{\alpha-1} \log	\frac{2}{\epsilon^2}\ .
\end{equation}
This combined with~\eqref{step3end} implies the claim
\begin{align}
\hmineps (X|\Theta K)_\rho \geq\quad&n\cdot \max_{s\in(0,1]} \frac{1}{s} \left[1+s-\log(1+2^s) \right]\nonumber\\
&-\frac{1}{s} \log \frac{2}{\epsilon^2}\ .
\end{align}

It is worth noting that as $n\rightarrow\infty$, the maximum is obtained for $s\rightarrow 0$, implying that as the system size approaches infinity, the optimal bound is still given by~\eqref{eq:bb84old}. That is, in terms of a bound which comes from the Shannon entropy. However, our analysis provides a better alternative to bound the smooth min-entropy for finite system sizes, and hence is more useful for practical implementations.


\section{Conclusions}

We have proven entropic uncertainty relations that pave the way for a practical implementation of BB84 and six-state protocols~\cite{chris:diss,serge:bounded,serge:new,Noisy1,noisy:robust,noisy:new,secureid,Curty10,secureid:practical} at small block length. Indeed, our relation has already been employed in~\cite{noisyimpl} for an experimental implementation of bit commitment in the bounded/noisy-storage model.


It is an interesting open question whether similarly strong relations can also be obtained with respect to quantum side information~\cite{Berta09,Coles11,qcextract}. This would allow security statements for such protocols in terms of the quantum capacity~\cite{qcextract} of the storage device, rather than the classical capacity~\cite{noisy:new} or the entanglement cost~\cite{entCost}. For the six-state case this has been done (implicitly) in~\cite{qcextract} for the special case of a R{\'e}nyi type entropy of order $\alpha = 2$, yielding however again a slightly weaker uncertainty relation as might be possible for other values of $\alpha \in (1,2]$. As the amount of uncertainty is the the key element in being able to tolerate experimental errors and losses in said protocols, it would be nice to extend our result to this setting.

\newpage

\onecolumngrid
\section*{Appendix}
In this appendix, we provide the technical details that lead to our claims. In section A, the complete proof for the uncertainty relation for BB84 bases (measurements in eigenstates of Pauli $\sigma_x$ and $\sigma_z$) is presented. In section B, similar methods are used to derive bounds for six-state bases (measurements in eigenstates of Pauli $\sigma_x$, $\sigma_y$ and $\sigma_z$).\\

We first restate the definitions of the relevant entropic quantities. Given any finite-dimensional Hilbert space $\hil$, let $\mathcal{S}_{\leq}(\hil)$ denote the set of sub-normalized density operators on $\hil$, and $\mathcal{S}(\hil)$ denote the set of normalized density operators on $\hil$. For $\hil_A$ and $\hil_B$, the conditional min-entropy of $\rho_{AB}\in\mathcal{S}(\hil_{A}\otimes\hil_{B})$ given $\sigma_B\in\mathcal{S}(\hil_B)$ is defined as
\begin{equation}
\hmin(A|B)_{\rho|\sigma} := \sup~\lbrace\lambda\in\mathbb{R}:2^{-\lambda}\cdot\id_A\otimes\sigma_B\geq\rho_{AB}\rbrace\ ,
\end{equation}
and the conditional min-entropy of $A$ given $B$ is defined as
\begin{equation}
\hmin(A|B)_{\rho} := \sup_{\sigma_{B}\in\mathcal{S}(\hil_{B})}~\hmin(A|B)_{\rho|\sigma}\ .
\end{equation}
The smooth conditional min-entropy of $A$ given $B$ and $\eps\geq0$ is defined as
\begin{align}
\hmineps(A|B)_{\rho} := \sup_{\rho'\in\mathcal{B}^{\eps}(\rho)}~\hmin(A|B)_{\rho'}\ ,
\end{align}
where $\mathcal{B}^{\eps}(\rho_{AB}):=\left\{\rho_{AB}'\in\mathcal{S}_{\leq}(\hil_{A}\otimes\hil_{B})|P=\sqrt{1-F^{2}(\rho,\rho')}\leq\eps\right\}$ is an $\eps$-ball in terms of the purified distance with
\begin{align}
F(\rho,\rho'):=\|\sqrt{\rho}\sqrt{\rho'}\|_{1}+\sqrt{(1-\mathrm{tr}[\rho])(1-\mathrm{tr}[\rho'])}
\end{align}
the (generalized) fidelity~\cite{Tomamichel09}.\\

The conditional $\alpha$-R{\'e}nyi entropies are defined as
\begin{equation}
\halpha(A|B)_{\rho|\rho} := \frac{1}{1-\alpha} \log\tr\left[\rho_{AB}^\alpha ( \id_A \otimes \rho_{B})^{1-\alpha}\right]\ ,
\end{equation}
where (possible) inverses are understood as generalized inverses. Note that there exist also slightly different definitions of conditional $\alpha$-R{\'e}nyi entropies in the literature.


\section{A. Uncertainty relation for BB84 measurements}\label{formalproof}

\subsubsection{Step 1 : Single qubit relation}

For any qubit state $\rho\in\mathcal{S}(\mathbb{C}^{2})$ we have to examine the quantities 
\begin{eqnarray}\label{alpharenyi}
\halpha (X|\Theta)_{\rho|\rho} &=& \frac{1}{1-\alpha} \log~P_\alpha (X|\Theta)\nonumber\\
P_\alpha (X|\Theta) &=& \tr\left[\rho_{X\Theta}^\alpha ( \id_X \otimes \rho_{\Theta})^{1-\alpha}\right]\nonumber\\
\rho_{X\Theta}&=&\sum_{\theta,x}p_{\theta}\cdot p_{x|\theta}\proj{x}\otimes\proj{\theta}\nonumber\\
p_{x|\theta}&=&\mathrm{tr}(M_{x|\theta}\rho)\ ,
\end{eqnarray}
with $M_{x|\theta}=\mathbf{H}^{\theta}\proj{x}\mathbf{H}^{\theta}$, and $\mathbf{H}=\frac{1}{\sqrt{2}}\begin{pmatrix}1&1\\1&-1\end{pmatrix}$ the Hadamard matrix. Since the choice of measurements is uniform, we get
\begin{align}
P_{\alpha}(X|\Theta)=\frac{1}{2}\cdot\sum_{\theta,x}p^{\alpha}_{x|\theta}\ .
\end{align}


\begin{theoremApp}\label{theorem1}
Let $\rho\in\mathcal{S}(\mathbb{C}^{2})$, and $\alpha=1+s$ with $s\in(0,1]$. Then we have for BB84 measurements as in~\eqref{alpharenyi} that
\begin{equation}\label{minimal_renyi_BB84}
\halpha(X|\Theta)_{\rho|\rho} \geq \frac{1}{s} [1+s-\log(1+2^s)]\ .
\end{equation}
\end{theoremApp}

\begin{proof}
We evaluate the term 
\begin{eqnarray}
P_{1+s}(X|\Theta) &=& \frac{1}{2}\cdot\sum_{\theta \in \{ 0,1 \} }  \sum_{x \in \{ 0,1 \}} p_{x|\theta}^{1+s} \nonumber \\
& & = \frac{1}{2} \left[ \tr(\rho |0\rangle\langle 0|)^{1+s}+\tr(\rho |1\rangle\langle 1|)^{1+s}+\tr(\rho |+\rangle\langle +|)^{1+s}+\tr(\rho |-\rangle\langle -|)^{1+s} \right] \nonumber \\
& & = \frac{1}{2^{2+s}} \lbrace [1+\tr(\sigma_z\rho )]^{1+s}+[1-\tr(\sigma_z\rho)]^{1+s}+[1+\tr(\sigma_x\rho)]^{1+s}+[1-\tr(\sigma_x\rho)]^{1+s}  \rbrace \nonumber\\
& & = \frac{1}{2^{2+s}} [ (1+z)^{1+s}+(1-z)^{1+s}+(1+x)^{1+s}+(1-x)^{1+s} ]\ ,
\end{eqnarray}
where $x:=\tr(\sigma_x\rho)$ and $z:=\tr(\sigma_z\rho)$. For any one qubit state $\rho$, we have the Bloch sphere condition\begin{equation}\label{anticomm}
\tr(\sigma_x\rho)^2 + \tr(\sigma_y\rho)^2 + \tr(\sigma_z\rho)^2 \leq1\ .
\end{equation}
and can therefore parametrize $x$ and $z$ by polar coordinates
\begin{equation}
x=r\sin \phi, z=r\cos \phi\ ,
\end{equation}
where $r \in [0,1]$, and $\phi \in [0,\frac{\pi}{2}]$. $P_\alpha(X|\Theta)$ can then be rewritten as a function depending on the variables $s$, $r$, and $\phi$
\begin{equation}
Q(s,r,\phi)=\frac{1}{2^{2+s}} [ (1+r\cos\phi)^{1+s}+(1-r\cos\phi)^{1+s}+(1+r\sin\phi)^{1+s}+(1-r\sin\phi)^{1+s} ]\ .
\end{equation}

The partial differential of $Q(s,r,\phi)$ with respect to $r$ becomes
\begin{equation}
\frac{\partial Q(s,r,\phi)}{\partial r} = \frac{1+s}{2^{2+s}}  [ \cos\phi(1+r\cos\phi)^{s}-\cos\phi(1-r\cos\phi)^{s}+\sin\phi(1+r\sin\phi)^{s}-\sin\phi(1-r\sin\phi)^{s} ]\ .
\end{equation}
Since in the range of $\phi$, $\sin\phi$ and $\cos\phi$ are positive, we obtain $\frac{\partial Q(s,r,\phi)}{\partial r} \geq 0$, which implies that the maximum is attained at $r=1$. The partial differential of $Q(s,r,\phi)$ with respect to $\phi$ at $r=1$ becomes
\begin{eqnarray}\label{firstderphi}
\frac{\partial Q(s,1,\phi)}{\partial \phi} &=& \frac{1+s}{2^{2+s}}  [ -\sin\phi(1+\cos\phi)^{s}+\sin\phi(1-\cos\phi)^{s}+\cos\phi(1+\sin\phi)^{s}-\cos\phi(1-\sin\phi)^{s} ] \nonumber\\
&=& \frac{1+s}{2^{2+s}} \lbrace\sin\phi[(1-\cos\phi)^s-(1+\cos\phi)^s]+\cos\phi[(1+\sin\phi)^s-(1-\sin\phi)^s]\rbrace\ .
\end{eqnarray}
For a stationary point of $Q(s,1,\phi)$, \eqref{firstderphi} is zero and the solution is obtained at three points: $\phi=0,~\frac{\pi}{4},~\frac{\pi}{2}$. The characteristics of the endpoints $\phi=0,\frac{\pi}{2}$ are the same, hence it suffices to analyze either. It remains to analyze the characteristic of these stationary points. To do so, we evaluate the second partial derivative at these points as a function of $s$
\begin{eqnarray}
f_1 (s) &=& \frac{\partial^2 Q(s,1,\phi)}{\partial \phi^2} \bigg|_{\phi=0} = \frac{1+s}{2^{1+s}} \left(s-2^{s-1}\right),\quad s\geq0 \\
f_2 (s) &=& \frac{\partial^2 Q(s,1,\phi)}{\partial \phi^2} \bigg|_{\phi=\frac{\pi}{4}}\nonumber\\
&=& \frac{1+s}{2^{2+s}} \left\lbrace s\cdot \left[\left(1-\frac{1}{\sqrt{2}}\right)^{s-1}+\left(1+\frac{1}{\sqrt{2}}\right)^{s-1}\right]-\sqrt{2}\left[\left(1+\frac{1}{\sqrt{2}}\right)^{s}-\left(1-\frac{1}{\sqrt{2}}\right)^{s}\right]\right\rbrace\ .\label{f2s}
\end{eqnarray}
To determine if the stationary point is a local minima or maxima, we show the positivity/negativity of these functions over the interval $s\in(0,1]$. Note that $f_1(0)=-\frac{1}{4}$ and $f_1(1)=0$, while $f_1(s)$ is always increasing since $\frac{\partial f_1(s)}{\partial s} = \frac{s}{2^{1+s}}[2-(1+s)\ln 2] \geq 0$. Hence $f_1(s)$ is negative, implying the endpoints correspond to a local maxima. On the other hand, note that the second term in~\eqref{f2s} is exactly of the form $g(a,s)$ as stated in Lemma C.\ref{positive} with $a=\frac{1}{\sqrt{2}}$. With this, we conclude that the point $\phi=\frac{\pi}{4}$ is a local minimum. This leaves the endpoints as the only candidates for optimal parameters that achieve the maxima of $Q(s,1,\phi)$. Evaluating $Q(s,1,0)$ then provides us the bound
\begin{equation}
P_{1+s}(X|\Theta) \leq  Q(s,1,0) = \frac{1}{2^{1+s}} (2^{s}+1)\ ,
\end{equation}
and plugging this back into \eqref{alpharenyi} gives \eqref{minimal_renyi_BB84}.
\end{proof}


\subsubsection{Step 2 : Relation for \texorpdfstring{$n$}{n}-qubits}

The goal is to prove that for any $n$-qubit state measured independently on each qubit in BB84 bases, the minimal output $\alpha$-R{\'e}nyi entropy is additive. Let first $n=2$ with the first system denoted by $A$ and the second by $B$. We have
\begin{eqnarray}\label{purity}
P_\alpha(X_{A}X_{B}|\Theta_{A}\Theta_{B})&=& \sum_{\theta_A,\theta_B} p_{\theta_A,\theta_B} \sum_{x_A,x_B}p_{x_A,x_B|\theta_A,\theta_B}^\alpha \nonumber\\
&=&  \frac{1}{2}\cdot\sum_{x_A,\theta_A} p_{x_A|\theta_A}^\alpha \cdot \frac{1}{2} \sum_{x_B,\theta_B} p_{x_B|x_A,\theta_A,\theta_B}^\alpha
\end{eqnarray}
where  $p_{\Theta_B|\Theta_A}=p_{\Theta_B}$ and $p_{\Theta_B=0}=p_{\Theta_b=1}=1/2$. Now assume that we have a one qubit upper bound
\begin{align}\label{eq:additive}
\frac{1}{2}\cdot\sum_{x_A,\theta_A} p_{x_A|\theta_A}^\alpha\leq c
\end{align}
for $P_\alpha(X|\Theta)$. Note that the second summation term in \eqref{purity} corresponds to $P_\alpha(X|\Theta)$ of the single qubit density operator
\begin{equation}\label{sigmaaftermeas}
\sigma_{B}=\tr_A \left[ \frac{M_{x_A|\theta_A}~\rho_{AB}~M_{x_A|\theta_A}^\dagger}{\tr[M_{x_A|\theta_A}~\rho_{AB}~M_{x_A|\theta_A}^\dagger]} \right]\ ,
\end{equation} 
where $M_{x_A|\theta_A}=\mathbf{H}^{\theta_A}|x_A\rangle\langle x_A|\mathbf{H}^{\theta_A}\otimes\id_{B}$. Hence we have 
\begin{eqnarray}
P_\alpha(X_{A}X_{B}|\Theta_{A}\Theta_{B})&\leq& \frac{c}{2}\cdot\sum_{x_A,\theta_A} p_{x_A|\theta_A}^\alpha\leq c^{2}\ .
\end{eqnarray}
The following lemma generalizes this argument to arbitrary $n$.

\begin{lemmaApp} For $\rho\in\mathcal{S}((\mathbb{C}^{2})^{\otimes n})$ measured independently on each qubit in BB84 bases, the minimal conditional $\alpha$-R{\'e}nyi entropy of $X^n$ with respect to $\Theta^n$ is additive.

\begin{proof} Consider
\begin{eqnarray}
P_\alpha(X^n|\Theta^n)_{\rho|\rho} &=& \sum_{\theta^n\in\lbrace0,1\rbrace^n} p_{\theta^n} \sum_{x^n\in\lbrace0,1\rbrace^n} p^\alpha_{x^n|\theta^n} \nonumber\\
&=& \frac{1}{2^n} \cdot\sum_{\theta^n\in\lbrace0,1\rbrace^n} \sum_{x^n\in\lbrace0,1\rbrace^n} \left(\displaystyle\prod_{i=1}^n p_{i|x^{i-1},\theta^{i-1}}\right)^\alpha
\end{eqnarray}
where $p_{i|x^{i-1},\theta^{i-1}} = p_{x_i|X^{i-1}=x^{i-1},\Theta^{i-1}=\theta^{i-1},K=k}$ for $i\geq2$ and $p_1=p_{x_1|\theta_1,K=k}$. Assuming the same upper bound as in~\eqref{eq:additive} we get
\begin{eqnarray}
P_\alpha(X^n|\Theta^n)_{\rho|\rho} &=& \frac{1}{2^{n-1}} \sum_{\theta^n\in\lbrace0,1\rbrace^n} \sum_{x^n\in\lbrace0,1\rbrace^n}\left(\displaystyle\prod_{i=1}^{n-1} p_{i|x^{i-1},\theta^{i-1}}\right)^\alpha \cdot \frac{1}{2} p_{n|x^{n-1},\theta^{n-1}}^\alpha \nonumber\\
&\leq& c \cdot \frac{1}{2^{n-1}} \sum_{\Theta^{n-1},X^{n-1}\in\lbrace0,1\rbrace^{n-1}}\left(\displaystyle\prod_{i=1}^{n-1} p_i\right)^\alpha\leq c^{n}\ .
\end{eqnarray}
\end{proof}
\end{lemmaApp}

Combining this with the one qubit uncertainty relation derived before, we obtain the following.

\begin{corollaryApp}\label{bb84uncert}
For $\alpha=1+s$ with $s\in(0,1]$, and $\rho\in\mathcal{S}((\mathbb{C}^{2})^{\otimes n})$ measured independently on each qubit in BB84 bases, we have
\begin{equation}
\halpha (X^n|\Theta^n)_{\rho|\rho} \geq n\cdot\frac{1}{s}\left[1+s-\log(1+2^{s})\right]\ .
\end{equation}
\end{corollaryApp}


\subsubsection{Step 3 : Classical side information K}

In Corollary A.\ref{bb84uncert} we have obtained an uncertainty relation $\halpha (X^n|\Theta^n)_{\rho|\rho}$ for any $n$-qubit states $\rho$. But generally we want to consider $n$-qubit states $\rho_k$ labelled with classical information $K$, and we need to make a relation to the quantity $\halpha(X^{n}|\Theta^{n}K)_{\rho|\rho}$ for the state $\rho=\sum_k p_k \rho_k$. That is, the $\alpha$-R{\'e}nyi entropy is also conditioned on classical information $K$. This quantity is evaluated as
\begin{eqnarray}
\halpha (X^n|\Theta^n K)_{\rho|\rho} &=& \frac{1}{1-\alpha} \log \sum_{k} \sum_{\theta^n\in\01^n} p_{k,\theta^n} \sum_{x^n\in\01^n} p^\alpha_{x^n|\theta^n,k}\nonumber\\
&=& \frac{1}{1-\alpha} \log \sum_{k} p_k \sum_{\theta^n\in\01^n} p_{\theta^n|k} \sum_{x^n\in\01^n} p^\alpha_{x^n|\theta^n,k}\ ,
\end{eqnarray}
where the difference is that now $p(\Theta|K=k)$ is conditioned on the classical information $K=k$. However, in our case $\Theta^n$ is chosen randomly regardless of what state is prepared. Thus $p(\Theta^n|K=k)=p(\Theta^n)=2^{-n}$ and we get
\begin{eqnarray}
\halpha (X^n|\Theta^{n}K)_{\rho|\rho} &=& \frac{1}{1-\alpha} \log \sum_{k} p_k \sum_{\theta^n\in\01^n} p_{\theta^n} \sum_{x^n\in\01^n} p^\alpha_{x^n|\theta^n,k}\nonumber\\
&\geq &n\cdot\frac{1}{s}\left[1+s-\log(1+2^{s})\right]\ .
\end{eqnarray}


\subsubsection{Step 4 : Relation to the min-entropy}

After obtaining a bound on $\halpha(X^n|\Theta^{n}K)_{\rho|\rho}$, we now link this to a bound on $\hmineps(X^n|\Theta^{n}K)_{\rho}$. It is shown in~\cite[Theorem 7]{QAEP} that for $\rho_{AB}\in \mathcal{S}(\mathcal{H}_{AB})$, $\epsilon\geq0$, and $\alpha\in(1,2]$
\begin{equation}
\hmineps (A|B)_{\rho} \geq\halpha(A|B)_{\rho|\rho} -\frac{1}{\alpha-1}\log\frac{2}{\epsilon^2}\ .
\end{equation}
Thus the smooth conditional min-entropy is lower bounded by general conditional $\alpha$-R{\'e}nyi entropies, with a correction term growing logarithmically in $1/\epsilon^2$. For the Shannon entropy ($\alpha\rightarrow1$) this term diverges, but considering $\alpha\in(1,2]$ the bound is very useful. Namely, the smooth conditional min-entropy of $X^n$ given $\Theta^{n}K$ is bounded to
\begin{equation}\label{finalbb84}
\frac{1}{n} ~\hmineps (X^n|\Theta^n K)_\rho \geq ~\frac{1}{n}H_{\alpha}(X^n|\Theta^n K)_{\rho|\rho} \geq \max_{s\in(0,1]}~ \frac{1}{s} [1+s-\log(1+2^s)] - \frac{1}{sn} \log \frac{2}{\epsilon^2}\ .
\end{equation}

Note that the maximum value of \eqref{finalbb84} is obtained for different values of $s$, as $n$ and $\epsilon$ varies.


\section{B: Uncertainty relation for six-state measurements}\label{sixstate}

In this section, we make use of the same methods as in Appendix A. We derive an uncertainty relation for any $n$-qubit state measured independently on each qubit in six-state bases. For the single qubit version, we have to consider
\begin{eqnarray}\label{alpharenyi2}
\halpha (X|\Theta)_{\rho|\rho} &=& \frac{1}{1-\alpha} \log~P_\alpha (X|\Theta)\nonumber\\
P_\alpha (X|\Theta) &=& \tr\left[\rho_{X\Theta}^\alpha ( \id_X \otimes \rho_{\Theta})^{1-\alpha}\right]\nonumber\\
\rho_{X\Theta}&=&\frac{1}{3}\cdot\sum_{\theta,x}p_{x|\theta}\proj{x}\otimes\proj{\theta}\nonumber\\
p_{x|\theta}&=&\mathrm{tr}(N_{x|\theta}\rho)\ ,
\end{eqnarray}
with $N_{x|\theta}=\mathbf{T}^{\theta}\proj{x}\mathbf{T}^{\theta}$, and $\mathbf{T}=\frac{1}{\sqrt{2}}\begin{pmatrix}1&-i\\1&i\end{pmatrix}$ the matrix that cyclically permutes the eigenbases of Pauli $\sigma_x$, $\sigma_y$, and $\sigma_z$.

\begin{theoremAppB}\label{theorem2}
Let $\rho\in\mathcal{S}(\mathbb{C}^{2})$, and $\alpha=1+s$ with $s\in(0,1]$. Then we have for six-state measurements as in~\eqref{alpharenyi2} that
\begin{equation}\label{minimal_renyi_6}
\halpha(X|\Theta)_{\rho|\rho} \geq\frac{-1}{s} \log\left[\frac{1}{3}(1+2^{1-s})\right]\ .
\end{equation}

\begin{proof}
We evaluate the term 
\begin{eqnarray}\label{sixstP}
P_{1+s}(X|\Theta) &=& \frac{1}{3}\cdot\sum_{x \in \{ 0,1 \} }  \sum_{\theta \in \{ 0,1,2 \}} p_{x|\theta}^{1+s} \nonumber \\
& & = \frac{1}{3}\cdot\frac{1}{2^{1+s}} \sum_{i=0}^2[(1+x_i)^{1+s}+(1-x_i)^{1+s}]\ ,
\end{eqnarray}
where $\lbrace x_0,x_1,x_2\rbrace:=\lbrace x,y,z\rbrace$ and $x_i:=\tr(\sigma_{x_i}\rho)$. Parametrizing this in terms of spherical coordinates, we write
\begin{equation}
x_0 = r\sin\phi\sin\theta, ~~~~x_1 = r\cos\phi\sin\theta, ~~~~ x_2 = r\cos\theta\ ,
\end{equation}
where $0\leq r \leq 1$, $0\leq \phi,\theta \leq \frac{\pi}{2}$. The expression~\eqref{sixstP} can be rewritten in terms of these new coordinates as
\begin{equation}
M(s,r,\phi,\theta):= \frac{1}{3}\cdot\frac{1}{2^{1+s}} \lbrace\sum_{p=0,1}\left[1+(-1)^p r\sin\phi\sin\theta\right]^{1+s}+\sum_{p=0,1}\left[1+(-1)^pr\cos\phi\sin\theta\right]^{1+s}+\sum_{p=0,1}\left[1+(-1)^p\cos\theta\right]^{1+s}\rbrace\ .
\end{equation}

Evaluating the partial differential of $Q(s,r,\phi)$ with respect to $r$ 
\begin{eqnarray}\nonumber
\frac{\partial M(s,r,\phi,\theta)}{\partial r} &=& \frac{1+s}{3}\cdot\frac{1}{2^{1+s}} \sin\theta \cdot \lbrace \sin\phi[(1+r\sin\phi\sin\theta)^{s}-(1-r\sin\phi\sin\theta)^{s}]\\
&&+\cos\phi[(1+r\cos\phi\sin\theta)^{s}-(1-r\cos\phi\sin\theta)^{s}] \rbrace\ .
\end{eqnarray}

Again we see that since in the range of $\phi,\theta$, all values of sines and cosines are positive, we obtain $\frac{\partial M(s,r,\phi,\theta)}{\partial r} \geq 0$, which implies the maximum is attained at $r=1$. Subsequently, evaluating the partial derivative
\begin{eqnarray}\nonumber
\frac{\partial M(s,1,\phi,\theta)}{\partial \phi} &=& \frac{1+s}{3}\cdot\frac{1}{2^{1+s}} \sin\theta \cdot \lbrace \cos\phi[(1+\sin\phi\sin\theta)^{s}-(1-\sin\phi\sin\theta)^{s}]\\
&&-\sin\phi[(1+r\cos\phi\sin\theta)^{s}-(1-r\cos\phi\sin\theta)^{s}] \rbrace\ ,
\end{eqnarray}
gives the points $\phi=0,~\frac{\pi}{4},~\frac{\pi}{2}$ as solutions. We continue by evaluating the second partial derivative at these points
\begin{eqnarray}\nonumber
\frac{\partial^2 M(s,1,\phi,\theta)}{\partial \phi^2}\bigg|_{\phi=0} &=& \frac{1+s}{3} \cdot \frac{1}{2^{1+s}} \cdot \sin\theta [2s\sin\theta - [(1+\sin\theta)^s-(1-\sin\theta)^s]]\\
\frac{\partial^2 M(s,1,\phi,\theta)}{\partial \phi^2}\bigg|_{\phi=\frac{\pi}{4}} &=& \frac{1+s}{3}\cdot\frac{1}{2^{s}} \cdot c^2 \cdot \lbrace s \cdot \left[(1+c)^{s-1}+(1-c)^{s-1} \right]- \frac{1}{c} \left[(1+c)^{s}-(1-c)^{s}\right] \rbrace\ ,
\end{eqnarray}
where $c=\frac{\sin\theta}{\sqrt{2}}$. By expanding in Taylor's series, the first equation is negative for $s\in(0,1]$, whereas the second equation is positive. Hence the maximum is obtained at $\phi=0$. The last step is to evaluate 
\begin{eqnarray}\nonumber
\frac{\partial M(s,1,0,\theta)}{\partial \theta} &=& \frac{1+s}{3}\cdot\frac{1}{2^{1+s}} \sin\theta \cdot \lbrace \cos\phi[(1+\sin\phi\sin\theta)^{s}-(1-\sin\phi\sin\theta)^{s}]\\
&&-\sin\phi[(1+r\cos\phi\sin\theta)^{s}-(1-r\cos\phi\sin\theta)^{s}] \rbrace\ .
\end{eqnarray}
But then this is of similar form as~\eqref{firstderphi}, and thus the maxima is obtained at $\theta=0$. Evaluating $M(s,1,0,0)$ then results in the claim
\begin{equation}
P_{1+s}(X|\theta) \leq M(s,1,0,0) = \frac{1}{3}(1+2^{1-s})\ .
\end{equation}
\end{proof}
\end{theoremAppB}

The additivity of minimal entropy holds by using the same argument as in Step 2 of Appendix A. Namely, given a string divided into parts $A$ and $B$, where $B$ denotes a single qubit system, the uncertainty relation for $B$ holds for the state
\begin{equation}
\sigma_B =\tr_A \left[ \frac{N_{x_A|\theta_A}~\rho_{AB}~N_{x_A|\theta_A}^\dagger}{\tr[N_{x_A|\theta_A}~\rho_{AB}~N_{x_A|\theta_A}^\dagger]} \right]\ ,
\end{equation}
where $N_{x_A|\theta_A}=\mathbf{T}^{\theta_A}|x_A\rangle\langle x_A|\mathbf{T}^{\theta_A}\otimes\id_{B}$. By exactly the same arguments as in Steps 3 and 4 in Appendix A, the smooth conditional min-entropy of the string $X^n\in\01^n$ conditioned on the basis $\theta^n\in\lbrace0,2\rbrace^n$ and the classical side information $K$ can then be bounded by
\begin{equation}
\frac{1}{n} ~\hmineps (X^n|\Theta^n K)_\rho \geq ~\frac{1}{n} H_{\alpha}(X^n|\Theta^n K)_{\rho|\rho} \geq \max_{s\in(0,1]}~ \frac{-1}{s} \log\left[\frac{1}{3}(1+2^{1-s})\right]-\frac{1}{sn}\cdot\log\frac{2}{\epsilon^{2}}\ .
\end{equation}


\section{C: Technical Lemmas}\label{technical}

\begin{lemmaAppC}\label{positive}
Given the function $g:\mathbb{R}\times\mathbb{R} \rightarrow \mathbb{R}$,
\begin{equation}
g(a,s):= s\cdot [(1+a)^{s-1}+(1-a)^{s-1}]-\frac{1}{a}\cdot [(1+a)^s-(1-a)^s]\ .
\end{equation}
Then $g(a,s)\geq 0$ for $a\in[0,1)$ and $s\in(0,1]$.
\end{lemmaAppC}

\begin{proof}
Since $a$ lies within the convergence radius of the function $(1\pm a)^s$, we expand the function in Taylor's series
\begin{align}
&s\cdot [(1+a)^{s-1}+(1-a)^{s-1}]-\frac{1}{a}(1+a)^s-(1-a)^s]\nonumber\\
&= 2s\cdot \left[1+\sum_{n=2,4...}\frac{(s-1)(s-2)...(s-n)}{n!}a^n\right]-\frac{1}{a} \left[2as+2\sum_{n=3,5...}\frac{s(s-1)...(s-n+1)}{n!} a^n\right]\nonumber\\
&= 2s\cdot \left[\sum_{n=2,4...}\frac{(s-1)(s-2)...(s-n)}{n!}a^n-\sum_{n=3,5...}\frac{(s-1)(s-2)...(s-n+1)}{n!} a^{n-1}\right]\nonumber\\
&= 2s\cdot \left[\sum_{n=2,4...}\frac{(s-1)(s-2)...(s-n)}{n!}a^n-\sum_{j=2,4...}\frac{(s-1)(s-2)...(s-j)}{(j+1)!} a^{j}\right]\nonumber\\
&= 2s\cdot \sum_{n=2,4...}(s-1)(s-2)...(s-n)~\frac{n}{(n+1)!}~a^n\nonumber\\
&\geq 0\ .
\end{align}
The first equality holds by a straightforward expansion of Taylor's series, the second equality by extracting $2s$ and absorbing $\frac{1}{a}$ into the second summation term, the third equality follows from redefining the summation variable $j=n-1$, and the last inequality follows because $(s-1)...(s-n)\geq 0$ for $s\in(0,1]$ and $n$ being an even integer.
\end{proof}

\end{document}